\documentclass[11pt]{article}

\usepackage[margin=1in]{geometry}
\usepackage{color}
\usepackage{amssymb}
\usepackage{amsmath}
\usepackage{amsthm}
\usepackage{url}
\usepackage{framed}
\usepackage{enumitem}

\newtheorem{theorem}{Theorem}
\newtheorem{definition}[theorem]{Definition}
\newtheorem{proposition}[theorem]{Proposition}
\newtheorem{lemma}[theorem]{Lemma}

\newtheorem{problem}[theorem]{Problem}

\begin{document}

\title{Robust width: A characterization of uniformly stable and robust compressed sensing}

\author{Jameson Cahill\footnote{Department of Mathematics, Duke University, Durham, NC}\qquad Dustin G.\ Mixon\footnote{Department of Mathematics and Statistics, Air Force Institute of Technology, Wright-Patterson AFB, OH}}

\maketitle

\begin{abstract}
Compressed sensing seeks to invert an underdetermined linear system by exploiting additional knowledge of the true solution.
Over the last decade, several instances of compressed sensing have been studied for various applications, and for each instance, reconstruction guarantees are available provided the sensing operator satisfies certain sufficient conditions.
In this paper, we completely characterize the sensing operators which allow uniformly stable and robust reconstruction by convex optimization for many of these instances.
The characterized sensing operators satisfy a new property we call the \textit{robust width property}, which simultaneously captures notions of \textit{widths} from approximation theory and of \textit{restricted eigenvalues} from statistical regression.
We provide a geometric interpretation of this property, we discuss its relationship with the restricted isometry property, and we apply techniques from geometric functional analysis to find random matrices which satisfy the property with high probability.
\end{abstract}

\section{Introduction}

Let $x^\natural$ be some unknown member of a finite-dimensional Hilbert space $\mathcal{H}$, and let $\Phi\colon\mathcal{H}\rightarrow\mathbb{F}^M$ denote some known linear operator, where $\mathbb{F}$ is either $\mathbb{R}$ or $\mathbb{C}$.
In a general form, \textit{compressed sensing} concerns the task of estimating $x^\natural$ provided
\begin{itemize}
\item[(i)]
we are told that $x^\natural$ is close in some sense to a particular subset $\mathcal{A}\subseteq\mathcal{H}$, and
\item[(ii)]
we are given data $y=\Phi x^\natural+e$ for some unknown $e\in\mathbb{F}^M$ with $\|e\|_2\leq\epsilon$.
\end{itemize}
Intuitively, if the subset $\mathcal{A}$ is ``small,'' then (i) offers more information about $x^\natural$, and so we might allow $M$ to be small, accordingly; we are chiefly interested in cases where $M$ can be smaller than the dimension of $\mathcal{H}$, as suggested by the name ``compressed sensing.''
A large body of work over the last decade has shown that for several natural choices of $\mathcal{A}$, there is a correspondingly natural choice of norm $\|\cdot\|_\sharp$ over $\mathcal{H}$ such that
\[
\Delta_{\sharp,\Phi,\epsilon}(y)
\quad
:=
\quad
\arg\min
\quad
\|x\|_\sharp
\quad
\mbox{subject to}
\quad
\|\Phi x-y\|_2\leq\epsilon
\]
is an impressively good estimate of $x^\natural$, provided the sensing operator $\Phi$ satisfies certain properties. 
However, the known sufficient conditions on $\Phi$ are not known to be necessary.

In this paper, we consider a broad class of triples $(\mathcal{H},\mathcal{A},\|\cdot\|_\sharp)$, and for each member of this class, we completely characterize the sensing operators $\Phi$ for which
\begin{equation}
\label{eq.guarantee}
\|\Delta_{\sharp,\Phi,\epsilon}(\Phi x^\natural+e)-x^\natural\|_2
\leq C_0\|x^\natural-a\|_\sharp+C_1\epsilon
\qquad
\forall x^\natural\in\mathcal{H},e\in\mathbb{F}^M,\|e\|_2\leq\epsilon,a\in \mathcal{A}
\end{equation}
for any given $C_0$ and $C_1$.
In the left-hand side above, we use $\|\cdot\|_2$ to denote the norm induced by the inner product over $\mathcal{H}$.
A comment on terminology:
Notice that the above guarantee is \textit{uniform} over all $x^\natural\in\mathcal{H}$.
Also, the $\|x^\natural-a\|_\sharp$ term ensures \textit{stability} in the sense that we allow $x^\natural$ to deviate from the signal model $\mathcal{A}$, whereas the $\epsilon$ term ensures \textit{robustness} in the sense that we allow for noise in the sensing process.

The next section describes how our main result fits with the current compressed sensing literature in the traditional sparsity case.
Next, Section~3 gives the main result: that \eqref{eq.guarantee} is equivalent to a new property we call the \textit{robust width property (RWP)}.
This guarantee holds for a variety of instances of compressed sensing, specifically, whenever the triple $(\mathcal{H},\mathcal{A},\|\cdot\|_\sharp)$ forms something we call a \textit{CS space}.
We identify several examples of CS spaces in Section~4 to help illustrate the extent of the generality.
In the special case where $\Phi\Phi^*=I$, the matrix $\Phi$ satisfies RWP precisely when a sizable neighborhood of its null space in the Grassmannian is contained in the set of null spaces of matrices which satisfy a natural generalization of the \textit{width property} in~\cite{KashinT:07}; we make this equivalence rigorous in Section~5.
In Section~6, we provide a direct proof that the restricted isometry property (RIP) implies RWP in the traditional sparsity case.
Section~7 then applies techniques from geometric functional analysis to show (without appealing to RIP) that certain random matrices satisfy RWP with high probability.
In fact, taking inspiration from~\cite{RaskuttiWY:10}, we produce a sensing matrix that satisfies RWP for the traditional sparsity case, but does not satisfy RIP, thereby proving that RIP is strictly stronger than \eqref{eq.rip guarantee}.
We conclude in Section~8 with some remarks.

\section{Perspective: The traditional case of sparsity}

Since the release of the seminal papers in compressed sensing~\cite{CandesRT:06,CandesT:05,Donoho:06}, the community has traditionally focused on the case in which $\mathcal{H}=\mathbb{R}^N$, $\mathcal{A}$ is the set $\Sigma_K$ of vectors with at most $K$ nonzero entries (hereafter, \textit{$K$-sparse vectors}), and $\|\cdot\|_\sharp$ is taken to be $\|\cdot\|_1$, defined by
\[
\|x\|_1:=\sum_{i=1}^N|x_i|.
\]
In this case, the \textit{null space property (NSP)} characterizes when $\Delta_{1,\Phi,0}(\Phi x^\natural)=x^\natural$ for every $K$-sparse $x^\natural$ (see Theorem~4.5 in~\cite{FoucartR:13}, for example).
Let $x_K$ denote the \textit{best $K$-term approximation of $x$}, gotten by setting all entries of $x$ to zero, save the $K$ of largest magnitude.
Then NSP states that
\begin{equation*}
\tag{$K$-NSP}
\|x\|_1<2\|x-x_K\|_1
\qquad
\forall x\in\operatorname{ker}(\Phi)\setminus\{0\}.
\end{equation*}
One may pursue some notion of stability by strengthening NSP.
Indeed,
\begin{equation*}
\tag{$(K,c)$-NSP}
\|x\|_1\leq c\|x-x_K\|_1
\qquad
\forall x\in\operatorname{ker}(\Phi)
\end{equation*}
is equivalent to
\begin{equation}
\label{eq.L1 stability}
\|\Delta_{1,\Phi,0}(\Phi x^\natural)-x^\natural\|_1
\leq\frac{2c}{2-c}\|x^\natural-x^\natural_K\|_1
\qquad
\forall x^\natural\in\mathbb{R}^N
\end{equation}
whenever $1<c<2$ (see Theorem 4.12 in~\cite{FoucartR:13}, for example).
Observe that $\|x^\natural-x^\natural_K\|_1\leq\|x^\natural-a\|_1$ for every $a\in\Sigma_K$, and so this is similar to \eqref{eq.guarantee} in the case where $\epsilon=0$.
However, we prefer a more isotropic notion of stability, and so we seek error bounds in $\ell_2$.
Of course, the estimate $\|u\|_2\leq\|u\|_1$ converts \eqref{eq.L1 stability} into such an error bound.
Still, we should expect to do better (by a factor of $\sqrt{K}$), as suggested by Theorem~8.1 in~\cite{CohenDD:09}, which gives that having
\begin{equation*}
\tag{$(2K,c)$-NSP$_{2,1}$}
\|x\|_2\leq \frac{c}{\sqrt{2K}}\|x-x_{2K}\|_1
\qquad
\forall x\in\operatorname{ker}(\Phi)
\end{equation*}
for some $c>0$ is equivalent to the existence a decoder $\Delta\colon\mathbb{R}^M\rightarrow\mathbb{R}^N$ such that
\begin{equation}
\label{eq.decoder guarantee}
\|\Delta(\Phi x^\natural)-x^\natural\|_2\leq\frac{C}{\sqrt{K}}\|x^\natural-x^\natural_K\|_1
\qquad
\forall x^\natural\in\mathbb{R}^N
\end{equation}
for some $C>0$.
(Observe that the $C$ in \eqref{eq.decoder guarantee} differs from the $C_0$ in \eqref{eq.guarantee} by a factor of $\sqrt{K}$; as one might expect, the $C_0$ in \eqref{eq.guarantee} depends on $\mathcal{A}$ in general.)
Unfortunately, the decoder $\Delta$ that~\cite{CohenDD:09} constructs for the equivalence is computationally inefficient.
Luckily, Kashin and Temlyakov~\cite{KashinT:07} remark that \eqref{eq.decoder guarantee} holds for $\Delta=\Delta_{1,\Phi,0}$ if and only if 
$\Phi$ satisfies the \textit{width property}:
\begin{equation*}
\tag{$(K,c)$-WP}
\|x\|_2\leq\frac{c}{\sqrt{K}}\|x\|_1
\qquad
\forall x\in\operatorname{ker}(\Phi)
\end{equation*}
for some $c>0$.
The name here comes from the fact that $c/\sqrt{K}$ gives the maximum radius (i.e., ``width'') of the portion of the unit $\ell_1$ ball that intersects $\operatorname{ker}(\Phi)$.
We note that while $(K,c)$-WP is clearly implied by $(2K,c)$-NSP$_{2,1}$ (and may appear to be a strictly weaker assumption), it is a straightforward exercise to verify that the two are in fact equivalent up to constants.
The moral here is that a sensing matrix $\Phi$ has a uniformly stable decoder $\Delta$ precisely when $\Delta_{1,\Phi,0}$ is one such decoder, thereby establishing that $\ell_1$ minimization is particularly natural for the traditional compressed sensing problem.

The robustness of $\ell_1$ minimization is far less understood.
Perhaps the most popular sufficient condition for robustness is the \textit{restricted isometry property}:
\begin{equation*}
\tag{$(2K,\delta)$-RIP}
(1-\delta)\|x\|_2^2
\leq\|\Phi x\|_2^2
\leq(1+\delta)\|x\|_2^2
\qquad
\forall x\in\Sigma_{2K}.
\end{equation*}
As the name suggests, an RIP matrix $\Phi$ acts as a near-isometry on $2K$-sparse vectors, and if $\Phi$ satisfies $(2K,\delta)$-RIP, then
\begin{equation}
\label{eq.rip guarantee}
\|\Delta_{1,\Phi,\epsilon}(\Phi x^\natural+e)-x^\natural\|_2\leq\frac{C_0}{\sqrt{K}}\|x^\natural-x^\natural_K\|_1+C_1\epsilon
\qquad
\forall x^\natural\in\mathbb{R}^N,e\in\mathbb{R}^M,\|e\|_2\leq\epsilon,
\end{equation}
where, for example, we may take $C_0=4.2$ and $C_1=8.5$ when $\delta=0.2$ (Theorem~1.2 in~\cite{Candes:08}).
As such, in the traditional sparsity case, RIP implies the condition \eqref{eq.guarantee} we wish to characterize.
Note that the lower RIP bound $\|\Phi x\|_2^2\geq(1-\delta)\|x\|_2^2$ is important since otherwise we might fail to distinguish a certain pair of $K$-sparse vectors given enough noise.
On the other hand, the upper RIP bound $\|\Phi x\|_2^2\leq(1+\delta)\|x\|_2^2$ is not intuitively necessary for \eqref{eq.rip guarantee}.

In pursuit of a characterizing property, one may be inclined to instead seek an NSP-type sufficient condition for \eqref{eq.rip guarantee}.
To this end, Theorem~4.22 in~\cite{FoucartR:13} gives that \eqref{eq.rip guarantee} is implied by the following robust version of the null space property:
\begin{equation*}
\tag{$(K,c_0,c_1)$-RNSP$_{2,1}$}
\|x_K\|_2\leq\frac{c_0}{\sqrt{K}}\|x-x_K\|_1+c_1\|\Phi x\|_2
\qquad
\forall x\in\mathbb{R}^N
\]
provided $0<c_0<1$ and $c_1>0$; in \eqref{eq.rip guarantee}, one may take $C_0=2(1+c_0)^2/(1-c_0)$ and $C_1=(3+c_0)c_1/(1-c_0)$.
This property appears to be a fusion of sorts between NSP and the lower RIP bound, which seems promising, but a proof of necessity remains elusive.

As a special case of the main result in this paper, we characterize the sensing matrices $\Phi$ which satisfy \eqref{eq.rip guarantee} as those which satisfy a new property called the \textit{robust width property}:
\begin{equation}
\tag{$(K,c_0,c_1)$-RWP}
\|x\|_2\leq \frac{c_0}{\sqrt{K}}\|x\|_1
\qquad
\forall x\in\mathbb{R}^N\mbox{ such that }\|\Phi x\|_2<c_1\|x\|_2.
\end{equation}
In particular, $C_0$ and $C_1$ scale like $c_0$ and $1/c_1$, respectively, in both directions of the equivalence (and with reasonable constants).
We note that $(K,c_0)$-WP follows directly from $(K,c_0,c_1)$-RWP.
Also, the contrapositive statement gives that $\|\Phi x\|_2\geq c_1\|x\|_2$ whenever $\|x\|_2>(c_0/\sqrt{K})\|x\|_1$, which not only captures the lower RIP bound we find appealing, but also bears some resemblance to the \textit{restricted eigenvalue property} that Bickel, Ritov and Tsybakov~\cite{BickelRT:09} use to produce guarantees for Lasso~\cite{Tibshirani:96} and the Dantzig selector~\cite{CandesT:07}:
\begin{equation}
\tag{$(K,c_0,c_1)$-REP}
\|\Phi x\|_2\geq c_0\|x\|_2
\qquad
\forall x\in\mathbb{R}^N\mbox{ such that }\|x-x_K\|_1\leq c_1\|x_K\|_1.
\end{equation}
Indeed, both RWP and REP impose a lower RIP--type bound on all nearly sparse vectors.
The main distinction between RWP and REP is the manner in which ``nearly sparse'' is technically defined.

\section{Main result}

In this section, we present a characterization of stable and robust compressed sensing.
This result can be applied to various instances of compressed sensing, and in order to express these instances simultaneously, it is convenient to make the following definition:

\begin{definition}
A \textit{CS space $(\mathcal{H},\mathcal{A},\|\cdot\|_\sharp)$ with bound $L$} consists of a finite-dimensional Hilbert space $\mathcal{H}$, a subset $\mathcal{A}\subseteq\mathcal{H}$, and a norm $\|\cdot\|_\sharp$ on $\mathcal{H}$ with the following properties:
\begin{itemize}
\item[(i)]
$0\in\mathcal{A}$.
\item[(ii)]
For every $a\in\mathcal{A}$ and $z\in\mathcal{H}$, there exists a decomposition $z=z_1+z_2$ such that
\[
\|a+z_1\|_\sharp=\|a\|_\sharp+\|z_1\|_\sharp,
\qquad
\|z_2\|_\sharp\leq L\|z\|_2.
\]
\end{itemize}
\end{definition}

The first property above ensures that $\mathcal{A}$ is not degenerate, whereas the second property is similar to the notion of \textit{decomposability}, introduced by Negahban et al.~\cite{NegahbanRWY:12}.
Note that one may always take $z_1=0$ and $z_2=z$, leading to a trivial bound $L=\sup\{\|z\|_\sharp/\|z\|_2:z\in\mathcal{H}\setminus\{0\}\}$.
Written differently, this bound satisfies $L^{-1}=\min\{\|z\|_2:z\in\mathcal{H},\|z\|_\sharp=1\}$, which is the smallest width of the unit $\sharp$-ball $B_\sharp$ (this is an example of a \textit{Gelfand width}).
As our main result will show, any substantial improvement to this bound will allow for uniformly stable and robust compressed sensing.
Such improvement is possible provided one may always decompose $z=z_1+z_2$ so that $z_2$ either has small $\ell_2$ norm or satisfies $\|z_2\|_\sharp\ll\|z_2\|_2$, which is to say that $z_2/\|z_2\|_\sharp$ lies in proximity to a ``pointy'' portion of $B_\sharp$.
However, we can expect such choices for $z_2$ to be uncommon, and so the set of all $x$ satisfying
\begin{equation}
\label{eq.pythagoras}
\|a+x\|_\sharp=\|a\|_\sharp+\|x\|_\sharp
\end{equation}
should be large, accordingly.
For the sake of intuition, denote the \textit{descent cone of $\|\cdot\|_\sharp$ at $a$} by
\[
\mathcal{D}
:=\{y:\exists t>0\mbox{ such that }\|a+ty\|_\sharp\leq\|a\|_\sharp\}.
\]
In the appendix, we show that if $\|v\|_\sharp=\|a\|_\sharp$ and $v-a$ generates a bounding ray of $\mathcal{D}$, then every nonnegative scalar multiple $x$ of $v$ satisfies \eqref{eq.pythagoras}.
As such, the set of $x$ satisfying \eqref{eq.pythagoras} is particularly large, for example, when $B_\sharp$ is \textit{locally conic} at $a/\|a\|_\sharp$, that is, the neighborhood of $a/\|a\|_\sharp$ in $B_\sharp$ is identical to the neighborhood of $0$ in $\mathcal{D}$ (when translated by $a/\|a\|_\sharp$).
Note that $B_1$ is locally conic at any sparse vector, whereas $B_2$ is nowhere locally conic.

For the record, we make no claim that ours is the ultimate definition of a CS space; indeed, our main result might be true for a more extensive class of spaces, but we find this definition to be particularly broad.
We demonstrate this in the next section with a series of instances, each of which having received considerable attention in the literature.

Next, we formally define our characterizing property:

\begin{definition}
We say a linear operator $\Phi\colon\mathcal{H}\rightarrow\mathbb{F}^M$ satisfies the \textit{$(\rho,\alpha)$-robust width property (RWP) over $B_\sharp$} if
\[
\|x\|_2\leq\rho\|x\|_\sharp
\]
for every $x\in\mathcal{H}$ such that $\|\Phi x\|_2<\alpha\|x\|_2$.
\end{definition}

If $\Phi$ satisfies $(\rho,\alpha)$-RWP, then its null space necessarily intersects the unit $\sharp$-ball $B_\sharp$ with maximum radius $\leq\rho$.
The RWP gets its name from this geometric feature; ``robust'' comes from the fact that points in $B_\sharp$ which are sufficiently close to the null space exhibit the same concentration in $\ell_2$.
We further study this geometric meaning in Section~5.
In the meantime, we give the main result:

\begin{theorem}
\label{thm.width-recovery}
For any CS space $(\mathcal{H},\mathcal{A},\|\cdot\|_\sharp)$ with bound $L$ and any linear operator $\Phi\colon\mathcal{H}\rightarrow\mathbb{F}^M$, the following are equivalent up to constants:
\begin{itemize}
\item[(a)]
$\Phi$ satisfies the $(\rho,\alpha)$-robust width property over $B_\sharp$.
\item[(b)]
For every $x^\natural\in\mathcal{H}$, $\epsilon\geq0$ and $e\in\mathbb{F}^M$ with $\|e\|_2\leq\epsilon$, any solution $x^\star$ to
\[
\arg\min
\quad
\|x\|_\sharp
\quad
\mathrm{subject~to}
\quad
\|\Phi x-(\Phi x^\natural+e)\|_2\leq\epsilon
\]
satisfies $\|x^\star-x^\natural\|_2\leq C_0\|x^\natural-a\|_\sharp+C_1\epsilon$ for every $a\in\mathcal{A}$.
\end{itemize}
In particular, (a) implies (b) with 
\[
C_0=4\rho,
\qquad
C_1=\frac{2}{\alpha}
\]
provided $\rho\leq1/(4L)$.
Also, (b) implies (a) with 
\[
\rho=2C_0,
\qquad
\alpha=\frac{1}{2C_1}.
\]
\end{theorem}

Notice that $C_0$ scales with $\rho$, while $C_1$ scales with $\alpha^{-1}$.
In the next section, we provide a variety of examples of CS spaces for which $L=O(\sqrt{K})$ for some parameter $K$, and since $C_0$ scales with $\rho=O(1/L)$, we might expect to find reconstruction guarantees for these spaces with $C_0=O(1/\sqrt{K})$; indeed, this has been demonstrated in~\cite{Candes:08,EldarM:09,MohanF:10,NeedellW:13b,NeedellW:13,RauhutW:14}.
Also, the fact that $C_1$ scales with $\alpha^{-1}$ is somewhat intuitive:
First, suppose that
\[
\|a\|_\sharp
\leq L\|a\|_2
\qquad
\forall a\in\mathcal{A}.
\]
(This occurs for every CS space considered in the next section.)
If $L<\rho^{-1}$, then the contrapositive of RWP gives that every point in $\mathcal{A}$ avoids the null space of $\Phi$.
The extent to which these points avoid the null space is captured by $\alpha$, and so we might expect more stability when $\alpha$ is larger (as is the case).

\begin{proof}[Proof of Theorem~\ref{thm.width-recovery}]
(b)$\Rightarrow$(a):
Pick $x^\natural$ such that $\|\Phi x^\natural\|_2<\alpha\|x^\natural\|_2$, and set $\epsilon=\alpha\|x^\natural\|_2$ and $e=0$.
Due to the feasibility of $x=0$, we may take $x^\star=0$, and so
\[
\|x^\natural\|_2
=\|x^\star-x^\natural\|_2
\leq C_0\|x^\natural\|_\sharp+C_1\epsilon
=C_0\|x^\natural\|_\sharp+\alpha C_1\|x^\natural\|_2,
\]
where the inequality applies (b) with $a=0$, which is allowed by property (i).
Isolating $\|x^\natural\|_2$ then gives
\[
\|x^\natural\|_2
\leq\frac{C_0}{1-\alpha C_1}\|x^\natural\|_\sharp
=\rho\|x^\natural\|_\sharp,
\]
where we take $\alpha=(2C_1)^{-1}$ and $\rho=2C_0$.

(a)$\Rightarrow$(b)
Pick $a\in\mathcal{A}$, and decompose $x^\star-x^\natural=z_1+z_2$ according to property (ii), i.e., so that $\|a+z_1\|_\sharp=\|a\|_\sharp+\|z_1\|_\sharp$ and $\|z_2\|_\sharp\leq L\|x^\star-x^\natural\|_2$.
Then
\begin{align*}
\|a\|_\sharp+\|x^\natural-a\|_\sharp
&\geq\|x^\natural\|_\sharp\\
&\geq\|x^\star\|_\sharp\\
&=\|x^\natural+(x^\star-x^\natural)\|_\sharp\\
&=\|a+(x^\natural-a)+z_1+z_2\|_\sharp\\
&\geq\|a+z_1\|_\sharp-\|(x^\natural-a)+z_2\|_\sharp\\
&\geq\|a+z_1\|_\sharp-\|x^\natural-a\|_\sharp-\|z_2\|_\sharp\\
&=\|a\|_\sharp+\|z_1\|_\sharp-\|x^\natural-a\|_\sharp-\|z_2\|_\sharp.
\end{align*}
Rearranging then gives $\|z_1\|_\sharp\leq2\|x^\natural-a\|_\sharp+\|z_2\|_\sharp$, which implies
\begin{equation}
\label{eq.general to bound 1}
\|x^\star-x^\natural\|_\sharp
\leq\|z_1\|_\sharp+\|z_2\|_\sharp
\leq2\|x^\natural-a\|_\sharp+2\|z_2\|_\sharp.
\end{equation}
Assume $\|x^\star-x^\natural\|_2>C_1\epsilon$, since otherwise we are done.
With this, we have
\[
\|\Phi x^\star-\Phi x^\natural\|_2
\leq\|\Phi x^\star-(\Phi x^\natural+e)\|_2+\|e\|_2
\leq2\epsilon
<2C_1^{-1}\|x^\star-x^\natural\|_2
=\alpha\|x^\star-x^\natural\|_2,
\]
where we take $C_1=2\alpha^{-1}$.
By (a), we then have 
\begin{equation}
\label{eq.applying a}
\|x^\star-x^\natural\|_2
\leq\rho\|x^\star-x^\natural\|_\sharp.
\end{equation}
Next, we appeal to a property of $z_2$:
\[
\|z_2\|_\sharp
\leq L\|x^\star-x^\natural\|_2
\leq \rho L\|x^\star-x^\natural\|_\sharp,
\]
where the last step follows from \eqref{eq.applying a}.
Substituting into \eqref{eq.general to bound 1} and rearranging then gives
\[
\|x^\star-x^\natural\|_\sharp
\leq\frac{2}{1-2\rho L}\|x^\natural-a\|_\sharp
\leq 4\|x^\natural-a\|_\sharp,
\]
provided $\rho\leq1/(4L)$.
Finally, we apply \eqref{eq.applying a} again to get
\[
\|x^\star-x^\natural\|_2
\leq \rho\|x^\star-x^\natural\|_\sharp
\leq 4\rho\|x^\natural-a\|_\sharp
=C_0\|x^\natural-a\|_\sharp
\leq C_0\|x^\natural-a\|_\sharp+C_1\epsilon,
\]
taking $C_0:=4\rho$.
\end{proof}

Notice that the above proof does not make use of every property of the norm $\|\cdot\|_\sharp$.
In particular, it suffices for $\|\cdot\|_\sharp\colon\mathcal{H}\rightarrow\mathbb{R}$ to satisfy
\begin{itemize}
\item[(i)]
$\|x\|_\sharp\geq\|0\|_\sharp$ for every $x\in\mathcal{H}$, and
\item[(ii)]
$\|x+y\|_\sharp\leq\|x\|_\sharp+\|y\|_\sharp$ for every $x,y\in\mathcal{H}$.
\end{itemize}
For example, one may take $\|\cdot\|_\sharp$ to be a \textit{seminorm}, i.e., a function which satisfies every norm property except positive definiteness, meaning $\|x\|_\sharp$ is allowed to be zero when $x$ is nonzero.
As another example, in the case where $\mathcal{H}=\mathbb{R}^N$, one may take 
\[
\|x\|_\sharp
=\|x\|_p^p
:=\sum_{i=1}^N|x_i|^p,
\qquad
0<p<1.
\]
This choice of objective function was first proposed by Chartrand~\cite{Chartrand:07}.
Since $B_\sharp\subseteq B_1$, then for any $\alpha$, we might expect $\Phi$ to satisfy $(\rho,\alpha)$-RWP over $B_\sharp$ with a smaller $\rho$ than for $B_1$.
As such, minimizing $\|x\|_\sharp$ instead of $\|x\|_1$ could very well yield more stability or robustness, though potentially at the price of computational efficiency since this alternative minimization is not a convex program.

\section{CS spaces}

In this section, we identify a variety of examples of CS spaces to illustrate the generality of our main result from the previous section.
For each example, we use the following lemma as a proof technique:

\begin{lemma}
Consider a finite-dimensional Hilbert space $\mathcal{H}$, subsets $\mathcal{A},\mathcal{B}\subseteq\mathcal{H}$, and a norm $\|\cdot\|_\sharp$ on $\mathcal{H}$ satisfying the following:
\begin{itemize}
\item[(i)]
$0\in\mathcal{A}$.
\item[(ii)]
For every $a\in\mathcal{A}$ and $z\in\mathcal{H}$, there exists a decomposition $z=z_1+z_2$ with $\langle z_1,z_2\rangle=0$ such that
\[
\|a+z_1\|_\sharp=\|a\|_\sharp+\|z_1\|_\sharp,
\qquad
z_2\in \mathcal{B}.
\]
\item[(iii)]
$\|b\|_\sharp\leq L\|b\|_2$ for every $b\in\mathcal{B}$.
\end{itemize}
Then $(\mathcal{H},\mathcal{A},\|\cdot\|_\sharp)$ is a CS space with bound $L$.
\end{lemma}

In words, the lemma uses an auxiliary set $\mathcal{B}$ to orthogonally decompose any $z\in\mathcal{H}$.
The conclusion follows form the fact that, since $z_2\in\mathcal{B}$ and $z_2$ is orthogonal to $z_1$, we have
\[
\|z_2\|_\sharp
\leq L\|z_2\|_2
\leq L\sqrt{\|z_1\|_2^2+\|z_2\|_2^2}
=L\|z\|_2.
\]
The remainder of this section uses this lemma to verify several CS spaces, as summarized in Table~\ref{table.cs spaces}.

\begin{table}
\caption{Examples of CS spaces.}\label{table.cs spaces}
\begin{center}
\begin{footnotesize}
\bgroup
\def\arraystretch{1.5}
\begin{tabular}{lccccc}
\hline
Instance & $\mathcal{H}$ & $\mathcal{A}$ & $\|x\|_\sharp$ & $L$ & Location \\
\hline
Weighted sparsity & $\mathbb{R}^N$ or $\mathbb{C}^N$ & $\Sigma_{W,K}$ & $\|Wx\|_1$ & $\sqrt{K}$ & Sec.~\ref{sec.1}\\
Block sparsity & $\bigoplus_{j\in J}\mathcal{H}_j$ & $B^{-1}(\Sigma_K)$ & $\|B(x)\|_1$ & $\sqrt{K}$ & Sec.~\ref{sec.2}\\
Gradient sparsity & $\operatorname{ker}(\nabla)^\perp$ & $\nabla^{-1}(\Sigma_K)$ & $\|\nabla x\|_1$ & $2\Delta\sqrt{K}$ & Sec.~\ref{sec.3}\\
Low-rank matrices & $\mathbb{R}^{m\times n}$ or $\mathbb{C}^{m\times n}$ & $\sigma^{-1}(\Sigma_K)$ & $\|x\|_*$ & $\sqrt{2K}$ & Sec.~\ref{sec.4}\\
\hline
\end{tabular}
\egroup
\end{footnotesize}
\end{center}
\end{table}

\subsection{Weighted sparsity}
\label{sec.1}

Take $\mathcal{H}$ to be either $\mathbb{R}^N$ or $\mathbb{C}^N$.
Let $W$ be a diagonal $N\times N$ matrix of positive weights $\{w_i\}_{i=1}^N$, and take $\|x\|_\sharp=\|Wx\|_1$ for every $x\in\mathcal{H}$.
The \textit{weighted sparsity} of a vector $x$ is defined to be $S_W(x):=\sum_{i\in\operatorname{supp}(x)}w_i^2$, and we take $\mathcal{A}=\mathcal{B}=\Sigma_{W,K}:=\{x:S_W(x)\leq K\}$.

To verify property (ii), fix any vector $a\in\mathcal{A}$.
Then for any vector $z$, pick $z_2$ to be the restriction of $z$ to the support of $a$, i.e.,
\[
z_2[i]:=
\left\{\begin{array}{cl}z[i]&\mbox{if }i\in\operatorname{supp}(a)\\0&\mbox{otherwise.}\end{array}\right.
\]
Also, pick $z_1:=z-z_2$.
Then $z_1$ and $z_2$ have disjoint support, implying $\langle z_1,z_2\rangle=0$.
Next, since $Wa$ and $Wz_1$ have disjoint support, we also have $\|W(a+z_1)\|_1=\|Wa\|_1+\|Wz_1\|_1$.
Finally, $S_W(z_2)=S_W(a)\leq K$, meaning $z_2\in\mathcal{B}$.
Note that we can take the bound of this CS space to be $L=\sqrt{K}$ since for every $b\in\mathcal{B}=\Sigma_{W,K}$, Cauchy--Schwarz gives
\[
\|b\|_\sharp
=\sum_{i=1}^Nw_i|b_i|
=\sum_{i\in\operatorname{supp}(b)}w_i|b_i|
\leq\bigg(\sum_{i\in\operatorname{supp}(b)}w_i^2\bigg)^{1/2}\|b\|_2
\leq \sqrt{K}\|b\|_2.
\]

Weighted sparsity has been applied in a few interesting ways.
If one is given additional information about the support of the desired signal $x^\natural$, he might weight the entries in the optimization to his benefit.
Suppose one is told of a subset of indices $T$ such that the support of $x^\natural$ is guaranteed to overlap with at least 10\% of $T$ (say).
For example, if $x^\natural$ is the wavelet transform of a natural image, then the entries of $x$ which correspond to lower frequencies tend to be large, so these indices might be a good choice for $T$.
When minimizing the $\ell_1$ norm of $x$, if we give less weight to the entries of $x$ over $T$, then the weighted minimizer will exhibit less $\ell_2$ error, accordingly~\cite{FriedlanderMSY:12,YuB:13}.
As another example, suppose $x^\natural$ denotes coefficients in an orthonormal basis over some finite-dimensional subspace $F\subseteq L^2([0,1])$, and suppose $y=\Phi x^\natural+e$ are noisy samples of the corresponding function at $M$ different points in $[0,1]$.
If one wishes to interpolate these samples with a smooth function that happens to be a sparse combination of basis elements in $F$, then he can encourage smoothness by weighting the entries appropriately~\cite{RauhutW:14}.
Uniformly stable and robust compressed sensing with weighted sparsity was recently demonstrated by Rauhut and Ward~\cite{RauhutW:14}.

\subsection{Block sparsity}
\label{sec.2}

Let $\{\mathcal{H}_j\}_{j\in J}$ be a finite collection of Hilbert spaces, and take $\mathcal{H}:=\bigoplus_{j\in J}\mathcal{H}_j$.
For some applications, it is reasonable to model interesting signals as $x\in\mathcal{H}$ such that, for most indices $j$, the component $x_j$ of the signal in $\mathcal{H}_j$ is zero; such signals are called \textit{block sparse}.
Notationally, we take $B\colon\bigoplus_{j\in J}\mathcal{H}_j\rightarrow\ell(J)$ to be defined entrywise by
\[
(B(x))[j]=\|x_j\|_2,
\]
and then write $B^{-1}(\Sigma_K)$ as the set of $K$-block sparse vectors.
(Here, $\ell(J)$ denotes the set of real-valued functions over $J$.)
In this case, we set $\mathcal{A}=\mathcal{B}=B^{-1}(\Sigma_K)$ and $\|x\|_\sharp=\|B(x)\|_1$.
Property (ii) then follows by an argument which is analogous to the weighted sparsity case.
To find the bound of this CS space, pick $b\in\mathcal{B}$ and note that $B(b)$ is $K$-sparse.
Then
\[
\|b\|_\sharp
=\|B(b)\|_1
\leq\sqrt{K}\|B(b)\|_2
=\sqrt{K}\|b\|_2,
\]
and so we may take $L=\sqrt{K}$.

Block sparsity has been used to help estimate multi-band signals, measure gene expression levels, and perform various tasks in machine learning (see~\cite{ElhamifarV:12} and references therein).
Similar to minimizing $\|B(x)\|_1$, Bakin~\cite{Bakin:99} (and more recently, Yuan and Lin~\cite{YuanL:06}) proposed the \textit{group lasso} to facilitate model selection by partitioning factors into blocks for certain real-world instances of the multifactor analysis-of-variance problem.
Eldar and Mishali~\cite{EldarM:09} proved that minimizing $\|B(x)\|_1$ produces uniformly stable and robust estimates of block sparse signals, and the simulations in~\cite{EldarKB:10} demonstrate that this minimization outperforms standard $\ell_1$ minimization when $x^\natural$ is block sparse.

\subsection{Gradient sparsity}
\label{sec.3}

Take $G=(V,E)$ to be a directed graph, and let $\ell(V)$ and $\ell(E)$ denote the vector spaces of real-valued (or complex-valued) functions over $V$ and $E$, respectively.
Then the \textit{gradient} $\nabla\colon\ell(V)\rightarrow\ell(E)$ is the linear operator defined by
\[
(\nabla x)[(i,j)]=x(j)-x(i)
\qquad
\forall(i,j)\in E
\]
for each $x\in\ell(V)$.
Take $\mathcal{H}=\operatorname{ker}(\nabla)^\perp$ and consider the \textit{total variation norm} $\|x\|_\sharp=\|\nabla x\|_1$.
We will verify property (ii) for $\mathcal{A}=\nabla^{-1}(\Sigma_K)$ and $\mathcal{B}=\nabla^{-1}(\Sigma_{2K\Delta})$, where $\Delta$ denotes the maximum total degree of $G$.

To this end, pick any $a$ such that $\nabla a$ is $K$-sparse, and let $H$ denote the subgraph $(V,\operatorname{supp}(\nabla a))$.
Take $C$ to be the subspace of all $x\in\mathcal{H}$ such that $x[i]=x[j]$ whenever $i$ and $j$ lie in a common (weak) component of $H$.
Then given any $z\in\mathcal{H}$, decompose $z$ using orthogonal projections: $z_1=P_Cz$ and $z_2=P_{C^\perp}z$.
We immediately have $z=z_1+z_2$ and $\langle z_1,z_2\rangle=0$.
Next, since $z_1\in C$, we know that $(i,j)\in\operatorname{supp}(\nabla a)$ implies $z_1[i]=z_1[j]$, and so $(\nabla z_1)[(i,j)]=0$.
As such, $\nabla a$ and $\nabla z_1$ have disjoint supports, and so
\[
\|\nabla(a+z_1)\|_1=\|\nabla a\|_1+\|\nabla z_1\|_1.\
\]
Finally, $H$ has $K$ edges, and so $H$ has at least $N-2K$ isolated vertices.
For each isolated vertex $i$, we have $z_1[i]=z[i]$, implying $z_2[i]=0$.
As such, $z_2$ is at most $2K$-sparse.
Since $(\nabla z_2)[(i,j)]$ is nonzero only if $z_2[i]$ or $z_2[j]$ is nonzero, then $\nabla z_2$ is only supported on the edges which are incident to support of $z_2$.
The easiest upper bound on this number of edges is $2K\Delta$, and so $z_2\in\mathcal{B}$.

We claim that this CS space has bound $L=2\Delta\sqrt{K}$.
To see this, first note that $\nabla^*\nabla=D-A$, where $D$ is the diagonal matrix of vertex total degrees, and where $A$ is the adjacency matrix of the underlying undirected graph.
Then
\[
\|\nabla\|_2^2
=\|\nabla^*\nabla\|_2
=\|D-A\|_2
\leq\|D\|_2+\|A\|_2
\leq 2\Delta,
\]
where the last step follows from Gershgorin's circle theorem.
As such, if $\nabla a$ is $2K\Delta$-sparse, we have
\[
\|\nabla a\|_1
\leq\sqrt{2K\Delta}\|\nabla a\|_2
\leq 2\Delta\sqrt{K}\|a\|_2.
\]
In particular, $L=O(\sqrt{K})$ when $\Delta$ is bounded.

One important example of gradient sparsity is total variation minimization for compressive imaging.
Indeed, if the pixels of an image are viewed as vertices of a grid graph (where each internal pixel has four neighbors: up, down, left, and right), then the total variation of the image $x$ is given by $\|x\|_\sharp$, which is often the objective function of choice (see~\cite{LustigDP:07}, for example).
It might also be beneficial to consider a $3$-dimensional image of voxels for applications like magnetic resonance imaging.
In either setting, the maximum degree is bounded ($\Delta=4,6$), and the uniform stability and robustness of compressed sensing in these settings has been demonstrated by Needell and Ward~\cite{NeedellW:13b,NeedellW:13}.

\subsection{Low-rank matrices}
\label{sec.4}

Let $\mathcal{H}$ be the Hilbert space of real (or complex) $n\times m$ matrices with inner product $\langle X,Y\rangle=\operatorname{Tr}[XY^*]$, and consider the \textit{nuclear norm} $\|X\|_\sharp=\|X\|_*$, defined to be the sum of the singular values of $X$.
Letting $\mathcal{A}=\sigma^{-1}(\Sigma_K)$ and $\mathcal{B}=\sigma^{-1}(\Sigma_{2K})$ denote the sets of matrices of rank at most $K$ and $2K$, respectively, then $L=\sqrt{2K}$, and property (ii) follows from a clever decomposition originally due to Recht, Fazel and Parrilo (Lemma~3.4 in~\cite{RechtFP:07}).

For any matrix $A$ of rank at most $K$, consider its singular value decomposition:
\[
A=U\left[\begin{array}{cc}\Sigma&0\\0&0\end{array}\right]V^*,
\]
where $\Sigma$ is $K\times K$.
Given any matrix $Z$, we then take $Y:=U^*ZV$ and consider the partition
\[
Y=\left[\begin{array}{cc}Y_{11}&Y_{12}\\Y_{21}&Y_{22}\end{array}\right],
\]
where $Y_{11}$ is $K\times K$.
We decompose $Z$ as the sum of the following matrices:
\[
Z_1:=U\left[\begin{array}{cc}0&0\\0&Y_{22}\end{array}\right]V^*,
\qquad
Z_2:=U\left[\begin{array}{cc}Y_{11}&Y_{12}\\Y_{21}&0\end{array}\right]V^*.
\]
It is straightforward to check that $\langle Z_1,Z_2\rangle=0$, and also that $AZ_1^*=0$ and $A^*Z_1=0$.
The latter conditions imply that $A$ and $Z_1$ have orthogonal row spaces and orthogonal column spaces, which in turn implies that $\|A+Z_1\|_*=\|A\|_*+\|Z_1\|_*$ (see Lemma~2.3 in~\cite{RechtFP:07}).
Finally, $[Y_{11}; Y_{21}]$ and $[Y_{12}; 0]$ each have rank at most $K$, and so $Z_2\in\mathcal{B}$.

This setting of compressed sensing also has a few applications.
For example, suppose you enter some signal into an unknown linear time-invariant system and then make time-domain observations.
If the system has low order, then its \textit{Hankel matrix} (whose entries are populated with translates of the system's impulse response) will have low rank.
As such, one can hope to estimate the system by minimizing the nuclear norm of the Hankel matrix subject to the observations (along with the linear constraints which define Hankel matrices)~\cite{RechtFP:07}.
For another application, consider \textit{quantum state tomography}, in which one seeks to determine a nearly pure quantum state (i.e., a low-rank self-adjoint positive semidefinite matrix with complex entries and unit trace) from observations which are essentially inner products with a collection of known matrices.
Then one can recover the unknown quantum state by minimizing the nuclear norm subject to the observations (and the linear constraint that the trace must be $1$)~\cite{GrossEtal:10}.
We note that the set $\mathbb{H}^{n\times n}$ of self-adjoint $n\times n$ matrices with complex entries is a real vector space of $n^2$ dimensions, which is slightly different from the setting of this subsection, but still forms a CS space with
\[
\mathcal{A}=\sigma^{-1}(\Sigma_K)\cap\mathbb{H}^{n\times n},
\qquad
\mathcal{B}=\sigma^{-1}(\Sigma_{2K})\cap\mathbb{H}^{n\times n},
\qquad
\|\cdot\|_\sharp=\|\cdot\|_*,
\qquad
L=\sqrt{2K}
\]
by the same proof.
For the original setting of not-necessarily-self-adjoint matrices, uniformly stable and robust compressed sensing was demonstrated by Mohan and Fazel~\cite{MohanF:10}.

\section{Geometric meaning of RWP}

The previous section characterized stable and robust compressed sensing in terms of a new property called the robust width property.
In this section, we shed some light on what this property means geometrically.
We start with a definition which we have adapted from~\cite{KashinT:07}:

\begin{definition}
We say a linear operator $\Phi\colon\mathcal{H}\rightarrow\mathbb{F}^M$ satisfies the \textit{$\rho$-width property over $B_\sharp$} if
\[
\|x\|_2\leq\rho\|x\|_\sharp
\]
for every $x$ in the null space of $\Phi$.
\end{definition}

Notice that the width property is actually a property of the null space of $\Phi$.
This is not the case for the robust width property since, for example, multiplying $\Phi$ by a scalar will have an effect on $\alpha$, and yet the null space remains unaltered.
In this section, we essentially mod out such modifications by focusing on a subclass of ``normalized'' sensing operators.
In particular, we only consider $\Phi$'s satisfying $\Phi\Phi^*=I$, which is to say that the measurement vectors are orthonormal.
For this subclass of operators, we will show that the robust width property is an intuitive property of the null space.
For further simplicity, we focus on the case in which $\mathcal{H}=\mathbb{R}^N$.

\begin{definition}
Let $\operatorname{Gr}(N,N-M)$ denote the Grassmannian, that is, the set of all subspaces of $\mathbb{R}^N$ of dimension $N-M$.
Given subspaces $X,Y\in\operatorname{Gr}(N,N-M)$, consider the corresponding orthogonal projections $P_X$ and $P_Y$.
Then the \textit{gap metric} $d$ over $\operatorname{Gr}(N,N-M)$ is defined by $d(X,Y):=\|P_X-P_Y\|_2$.
\end{definition}

\begin{theorem}
\label{thm.RWP vs WP}
A linear operator $\Phi\colon\mathbb{R}^N\rightarrow\mathbb{R}^M$ with $\Phi\Phi^*=I$ and null space $Y$ satisfies the $(\rho,\alpha)$-robust width property over $B_\sharp$ if and only if for every subspace $X$ with $d(X,Y)<\alpha$, every linear operator $\Psi\colon\mathbb{R}^N\rightarrow\mathbb{R}^M$ with $\Psi\Psi^*=I$ and null space $X$ satisfies the $\rho$-width property over $B_\sharp$.
\end{theorem}

Imagine the entire Grassmannian, and consider the subset $\mathrm{WP}$ corresponding to the subspaces satisfying the $\rho$-width property.
Then by Theorem~\ref{thm.RWP vs WP}, a point $Y\in\mathrm{WP}$ satisfies the $(\rho,\alpha)$-robust width property precisely when the entire open ball centered at $Y$ of radius $\alpha$ lies inside $\mathrm{WP}$.
One can also interpret this theorem in the context of compressed sensing.
First, observe that proving Theorem~\ref{thm.width-recovery} with $\epsilon=0$ gives that the width property is equivalent to stable compressed sensing (cf.~\cite{KashinT:07}).
As such, if $\Phi$ allows for stable and robust compressed sensing with robustness constant $C_1$, then every $\Psi$ whose null space is within $\alpha=(2C_1)^{-1}$ of the null space of $\Phi$ will necessarily enjoy stability.
The remainder of this section proves Theorem~\ref{thm.RWP vs WP} with a series of lemmas:

\begin{lemma}
\label{lemma.gap metric equivalence 1}
For any subspaces $X,Y\subseteq\mathbb{R}^N$, we have
\[
\min_{\substack{x\in X\\\|x\|_2=1}}\max_{\substack{y\in Y\\\|y\|_2=1}}\langle x,y\rangle
=\sqrt{1-\big(d(X,Y)\big)^2}.
\]
\end{lemma}

\begin{proof}
Pick $x\in X$ and $y\in Y$ with $\|x\|_2=\|y\|_2=1$.
Then Cauchy--Schwarz gives
\[
\langle x,y\rangle
=\langle P_{Y}x,y\rangle+\langle P_{Y^\perp}x,y\rangle
=\langle P_Yx,y\rangle
\leq\|P_Yx\|_2,
\]
with equality precisely when $y=P_Yx/\|P_Yx\|_2$.
As such, the left-hand side of the claimed identity can be simplified as
\[
\mathrm{LHS}
:=\min_{\substack{x\in X\\\|x\|_2=1}}\max_{\substack{y\in Y\\\|y\|_2=1}}\langle x,y\rangle
=\min_{\substack{x\in X\\\|x\|_2=1}}\|P_Yx\|_2.
\]
Next, we appeal to the Pythagorean theorem to get
\begin{equation}
\label{eq.deriving rhs}
\sqrt{1-\mathrm{LHS}^2}
=\max_{\substack{x\in X\\\|x\|_2=1}}\|P_{Y^\perp}x\|_2
=\max_{\substack{z\in\mathbb{R}^N\\\|z\|_2=1}}\|P_{Y^\perp}P_Xz\|_2
=\|P_{Y^\perp}P_X\|_2.
\end{equation}
Finally, we appeal to Theorem~2.6.1 in~\cite{GolubV:96}, which states that $d(X,Y)=\|P_{Y^\perp}P_X\|_2$.
Substituting into \eqref{eq.deriving rhs} and rearranging then gives the result.
\end{proof}

\begin{lemma}
\label{lemma.gap metric equivalence 2}
For any subspaces $X,Y\subseteq\mathbb{R}^N$, we have
\[
X\subseteq\{x:\|P_{Y^\perp}x\|_2<\alpha\|x\|_2\}
\]
if and only if for every $x\in X$, there exists $y\in Y$ of unit $\ell_2$ norm such that
\begin{equation}
\label{eq.large inner product}
\langle x,y\rangle
>\sqrt{1-\alpha^2}\|x\|_2.
\end{equation}
\end{lemma}

\begin{proof}
($\Rightarrow$)
Given $x\in X$, pick $y:=P_Y x/\|P_Y x\|_2$.
Then
\[
\langle x,y\rangle
=\bigg\langle x,\frac{P_Y x}{\|P_Y x\|_2}\bigg\rangle
=\|P_Y x\|_2
=\sqrt{\|x\|_2^2-\|P_{Y^\perp} x\|_2^2}
>\sqrt{1-\alpha^2}\|x\|_2,
\]
where the inequality follows from the assumed containment.

($\Leftarrow$)
Pick $x\in X$.
Then there exists $y\in Y$ of unit $\ell_2$ norm satisfying \eqref{eq.large inner product}.
Recall that for any subspaces $A\subseteq B$, the corresponding orthogonal projections satisfy $P_AP_B=P_A$.
Since $Y^\perp$ is contained in the orthogonal complement of $y$, we then have
\[
\|P_{Y^\perp}x\|_2^2
=\|P_{Y^\perp}P_{y^\perp}x\|_2^2
\leq\|P_{y^\perp}x\|_2^2
=\|x\|_2^2-|\langle x,y\rangle|^2
<\alpha^2\|x\|_2^2.
\]
Since our choice for $x$ was arbitrary, this proves the claim.
\end{proof}

We now use Lemmas~\ref{lemma.gap metric equivalence 1} and~\ref{lemma.gap metric equivalence 2} to prove the following lemma, from which Theorem~\ref{thm.RWP vs WP} immediately follows:

\begin{lemma}
Pick a subspace $Y\in\operatorname{Gr}(N,N-M)$.
Then
\[
\bigcup_{\substack{X\in\operatorname{Gr}(N,N-M)\\d(X,Y)<\alpha}}X
=\{x:\|P_{Y^\perp}x\|_2<\alpha\|x\|_2\}.
\]
\end{lemma}

\begin{proof}
Let $U$ and $E$ denote the left- and right-hand sides, respectively.
Lemma~\ref{lemma.gap metric equivalence 1} gives that $d(X,Y)<\alpha$ is equivalent to having that for every $x\in X$, there exists $y\in Y$ of unit $\ell_2$ norm such that $\langle x,y\rangle>\sqrt{1-\alpha^2}\|x\|_2$, which in turn is equivalent to $X\subseteq E$ by Lemma~\ref{lemma.gap metric equivalence 2}.
The fact that $d(X,Y)<\alpha$ implies $X\subseteq E$ immediately gives $U\subseteq E$. 
We will show that the converse implies the reverse containment, thereby proving set equality.

To this end, pick any $e\in E$.
If $e=0$, then $e\in Y$ lies in $U$.
Otherwise, we may assume $\|e\|_2=1$ without loss of generality since both $U$ and $E$ are closed under scalar multiplication.
Also, $U=\mathbb{R}^N=E$ if $\alpha>1$, and so we may assume $\alpha\leq1$.
Denote $Z:=Y\cap\operatorname{span}\{e\}^\perp$, and take $X:=Z+\operatorname{span}\{e\}$.
Since $\alpha\leq1$, we know that $\|P_{Y^\perp}e\|_2<\alpha\|e\|_2\leq\|e\|_2$, i.e., $e\not\in Y^\perp$, and so $\operatorname{dim}(Z)=\operatorname{dim}(Y)-1$, which in turn implies $X\in\operatorname{Gr}(N,N-M)$.
To see that $X\subseteq E$, pick any $x\in X$.
Then 
\[
P_{Y^\perp}x
=P_{Y^\perp}P_ex+P_{Y^\perp}P_Zx
=P_{Y^\perp}P_ex
=\langle x,e\rangle P_{Y^\perp}e.
\]
Taking norms of both sides then gives $\|P_{Y^\perp}x\|_2=\|P_{Y^\perp}e\|_2|\langle x,e\rangle|<\alpha\|x\|_2$, as desired.
Since $X\subseteq E$, we then have $d(X,Y)<\alpha$ by Lemmas~\ref{lemma.gap metric equivalence 1} and~\ref{lemma.gap metric equivalence 2}, and so $e\in X\subseteq U$.
Finally, since our choice for $e$ was arbitrary, we conclude that $E\subseteq U$.
\end{proof}

\section{A direct proof that RIP implies RWP}

Recall from Section~2 that the restricted isometry property (RIP) implies part (b) of Theorem~\ref{thm.width-recovery} in the traditional sparsity case~\cite{Candes:08}.
As such, one could pass through the equivalence to conclude that RIP implies the robust width property (RWP).
For completeness, this section provides a direct proof of this result.
(Instead of presenting different versions of RIP for different CS spaces, we focus on the traditional sparsity case here; the proofs for other cases are similar.)

For the direct proof, it is particularly convenient to use a slightly different version of RIP:
We say $\Phi$ satisfies the \textit{$(J,\delta)$-restricted isometry property} if
\[
(1-\delta)\|x\|_2
\leq\|\Phi x\|_2
\leq(1+\delta)\|x\|_2
\]
for every $J$-sparse vector $x$.
We note that this ``no squares'' version of RIP is equivalent to the original version up to constants, and is not unprecedented (see~\cite{RechtFP:07}, for example).

\begin{theorem}
\label{thm.rip implies rwp}
Suppose $\Phi$ satisfies the $(J,\delta)$-restricted isometry property with $\delta<1/3$.
Then $\Phi$ satisfies the $(\rho,\alpha)$-robust width property over $B_1$ with
\[
\rho=\frac{3}{\sqrt{J}},
\qquad
\alpha=\frac{1}{3}-\delta.
\]
\end{theorem}

The proof makes use of the following lemma:

\begin{lemma}
\label{lemma.rip trick}
Suppose $\Phi$ satisfies the right-hand inequality of the $(J,\delta)$-restricted isometry property.
Then for every $x$ such that $\|x\|_2>\rho\|x\|_1$, we have 
\[
\|x-x_J\|_2<\frac{1}{\rho\sqrt{J}}\|x\|_2,
\qquad
\|\Phi(x-x_J)\|_2<\frac{1+\delta}{\rho\sqrt{J}}\|x\|_2.
\]
\end{lemma}

\begin{proof}
Let $T_0$ denote the indices of the $J$ largest entries of $x$, and for each $j\geq 1$, let $T_j$ denote the indices of the $J$ largest entries of $x$ not covered by $T_i$ for $i<j$.
Then
\[
\|x_{T_{j+1}}\|_2
\leq\sqrt{J}\max_{i\in T_{j+1}}|x[i]|
\leq\sqrt{J}\min_{i\in T_{j}}|x[i]|
\leq\frac{1}{\sqrt{J}}\|x_{T_j}\|_1
\]
for every $j\geq0$.
As such,
\[
\|x-x_J\|_2
\leq\sum_{j\geq1}\|x_{T_j}\|_2
\leq\frac{1}{\sqrt{J}}\sum_{j\geq0}\|x_{T_j}\|_1
=\frac{1}{\sqrt{J}}\|x\|_1
<\frac{1}{\rho\sqrt{J}}\|x\|_2,
\]
where the last step uses the hypothesis.
Similarly,
\[
\|\Phi(x-x_J)\|_2
\leq\sum_{j\geq1}\|\Phi x_{T_j}\|_2
\leq(1+\delta)\sum_{j\geq1}\|x_{T_j}\|_2
<\frac{1+\delta}{\rho\sqrt{J}}\|x\|_2,
\]
as claimed.
\end{proof}

\begin{proof}[Proof of Theorem~\ref{thm.rip implies rwp}]
We will prove a contrapositive of sorts.
In particular, we will assume $\Phi$ satisfies the right-hand inequality of the restricted isometry property, but violates the robust width property, and we will show that $\Phi$ necessarily violates the left-hand inequality of the restricted isometry property.

To this end, pick $x$ such that $\|\Phi x\|_2<\alpha\|x\|_2$ but $\|x\|_2>\rho\|x\|_1$.
Then
\begin{equation}
\label{eq.rip to rwp 1}
\|\Phi x_J\|_2
\leq\|\Phi x\|_2+\|\Phi(x-x_J)\|_2
<\alpha\|x\|_2+\|\Phi(x-x_J)\|_2
<\frac{2}{3}(1-\delta)\|x\|_2,
\end{equation}
where the last step uses Lemma~\ref{lemma.rip trick} along with our choices for $\rho$ and $\alpha$.
Another application of Lemma~\ref{lemma.rip trick} gives
\[
\|x\|_2
\leq\|x_J\|_2+\|x-x_J\|_2
<\|x_J\|_2+\frac{1}{3}\|x\|_2.
\]
Isolating $\|x\|_2$ and substituting into \eqref{eq.rip to rwp 1} then gives $\|\Phi x_J\|_2<(1-\delta)\|x_J\|_2$, as desired.
\end{proof}

One nice feature of the above proof is that it reveals asymmetry between the upper and lower RIP bounds.
Indeed, the upper RIP bound is only used to control $\|\Phi(x-x_J)\|_2$, whereas the lower RIP bound is directly related to $\alpha$.
Notice that the contrapositive of RWP posits that
\begin{equation}
\label{eq.rwp contrapositive}
\|\Phi x\|_2\geq\alpha\|x\|_2
\end{equation}
whenever $\|x\|_2>\rho\|x\|_1$, which is the case when $x$ is $J$-sparse with $J\leq1/\rho^2$.
As such, RWP implies lower RIP with $\delta=1-\alpha$.
Going the other direction, we could take $\alpha=1-\delta$ and conclude from lower RIP that \eqref{eq.rwp contrapositive} holds for every sparse $x$, but in order to get \eqref{eq.rwp contrapositive} for the remaining ``nearly sparse'' vectors, namely, those for which $\|x\|_2>\rho\|x\|_1$, the above proof appeals to an upper RIP bound.
The next section presents some alternative methods for demonstrating RWP.

\section{RIP-free approaches to RWP}

Considering the previous section, one can certainly appeal to the restricted isometry property (RIP) in order to establish the robust width property (RWP) for a given sensing operator.
The purpose of this section is to provide alternative proof techniques for RWP.
Here, it will be convenient to use the contrapositive statement of RWP: 
A linear operator $\Phi\colon\mathcal{H}\rightarrow\mathbb{F}^M$ satisfies the $(\rho,\alpha)$-robust width property over $B_\sharp$ if and only if
\[
\|\Phi x\|_2\geq\alpha\|x\|_2
\]
for every $x\in\mathcal{H}$ such that $\|x\|_2>\rho\|x\|_\sharp$.
By scaling, RWP is further equivalent to having $\|\Phi x\|_2\geq\alpha$ for every $x$ with $\|x\|_2=1$ and $\|x\|_\sharp<\rho^{-1}$.
Since $\|\Phi x\|_2$ is a continuous function of $x$, we deduce the following lemma:

\begin{lemma}
\label{lemma.rwp characterization}
A linear operator $\Phi\colon\mathcal{H}\rightarrow\mathbb{F}^M$ satisfies the $(\rho,\alpha)$-robust width property over $B_\sharp$ if and only if $\|\Phi x\|_2\geq\alpha$ for every $x\in\rho^{-1}B_\sharp\cap\mathbb{S}$, where $\mathbb{S}$ denotes the unit sphere in $\mathcal{H}$.
\end{lemma}

In the following subsections, we demonstrate RWP by leveraging Lemma~\ref{lemma.rwp characterization} along with ideas from geometric functional analysis.
Interestingly, such techniques are known to outperform RIP-based analyses in certain regimes~\cite{BlanchardCT:11,RudelsonV:08}.
To help express how this section interacts with the remainder of the paper, we have included a ``cheat sheet for the practitioner'' that explains how one might apply this paper to future instances of structured sparsity.
The remainder of this section takes $\mathcal{H}=\mathbb{R}^N$ and $\mathbb{F}=\mathbb{R}$.

\begin{figure}[t]
\begin{framed}
\medskip
\begin{footnotesize}
\begin{center}
\uppercase{\textbf{Cheat sheet for the practitioner}}
\end{center}
Given an instance of structured sparsity for potential compressed sensing, apply the following process:
\begin{itemize}[leftmargin=0.8cm,rightmargin=0.4cm]
\item[1.]
\textbf{Verify that you have a CS space.}
Consult Section~4 for examples and proofs.  
\item[2.]
\textbf{Determine robust width parameters.}
Consult Theorem~\ref{thm.width-recovery} to determine the pairs $(\rho,\alpha)$ that make the $\ell_2$ error of reconstruction acceptably low.
\item[3.]
\textbf{Estimate a Gaussian width.}
For each acceptable $\rho$, estimate the Gaussian width of $\rho^{-1}B_\sharp\cap\mathbb{S}$.
This can be done in two ways:
\begin{itemize}
\item[(a)]
\textit{Analytically.}
Mimic the proof of Lemma~\ref{lemma.L1 width} or consult~\cite{LedouxT:91,MilmanS:86,Pisier:89}.
\item[(b)]
\textit{Numerically.}
Observe that $\sup_{x\in S}\langle x,g\rangle$ is a convex program for each $g$.
As such, take a random sample of iid $N(0,I)$ vectors $\{g_j\}_{j\in J}$, run the convex program for each $j\in J$, and produce an upper confidence bound on the parameter $w(S):=\mathbb{E}\sup_{x\in S}\langle x,g\rangle$.
\end{itemize}
\item[4.]
\textbf{Calculate the number of measurements.}
Depending on the type of measurement vectors desired, consult Proposition~\ref{thm.main random result}, Proposition~\ref{prop.bowling 1}, or more generally Theorem~6.3 in~\cite{Tropp:14}.
The number of measurements will depend on $\rho$ (implicitly through the Gaussian width) and on $\alpha$.
Minimize this number over the acceptable pairs $(\rho,\alpha)$.
\end{itemize}
\end{footnotesize}
\end{framed}
\end{figure}

\subsection{Gordon scheme}

In this subsection, as in Section~5, we focus on the case in which $\Phi$ satisfies $\Phi\Phi^*=I$.
Letting $Y$ denote the null space of $\Phi$, we then have $\|\Phi x\|_2=\|P_{Y^\perp}x\|_2=\operatorname{dist}(x,Y)$.
As such, by Lemma~\ref{lemma.rwp characterization}, $\Phi$ satisfies RWP if its null space is of distance at least $\alpha$ from $S:=\rho^{-1}B_\sharp\cap\mathbb{S}^{N-1}$.
In other words, it suffices for the null space to have empty intersection with an $\alpha$-thickened version of $S$.
As we will see, a random subspace of sufficiently small dimension will do precisely this.

The following theorem is originally due to Gordon (see Theorem~3.3 in~\cite{Gordon:88}); this particular version is taken from~\cite{Mixon:14} (namely, Proposition~2).
A few definitions are needed before stating the theorem.
Define the \textit{Gaussian width} of $S\subseteq\mathbb{R}^k$ to be
\[
w(S):=\mathbb{E}\sup_{x\in S} \langle x,g\rangle,
\]
where $g$ has iid $N(0,1)$ entries.
Also, we take $a_k$ to denote the Gaussian width of $\mathbb{S}^{k-1}$, i.e., $a_k:=\mathbb{E}\|g\|_2$.
Overall, the Gaussian width is a measure of the size of a given set, and its utility is illustrated in the following theorem:

\begin{theorem}[Escape through a thickened mesh]
\label{thm.escape}
Take a closed subset $S\subseteq\mathbb{S}^{N-1}$.
If $w(S)<(1-\epsilon)a_M-\epsilon a_N$, then a subspace $Y$ drawn uniformly from $\operatorname{Gr}(N,N-M)$ satisfies
\[
\operatorname{Pr}\Big(Y\cap(S+\epsilon B_2)=\emptyset\Big)
\geq1-\frac{7}{2}\operatorname{exp}\bigg(-\frac{1}{2}\bigg(\frac{(1-\epsilon)a_M-\epsilon a_N-w(S)}{3+\epsilon+\epsilon a_N/a_M}\bigg)^2\bigg).
\]
\end{theorem}

It is well known that $(1-1/k)\sqrt{k}<a_k<\sqrt{k}$ for each $k$.
Combined with Lemma~\ref{lemma.rwp characterization} and Theorem~\ref{thm.escape}, this quickly leads to the following result:

\begin{proposition}
\label{thm.main random result}
Fix $\lambda=M/N$, pick $\alpha<1/(1+1/\sqrt{\lambda})$, and denote $S=\rho^{-1}B_\sharp\cap\mathbb{S}^{N-1}$.
Suppose
\[
M\geq C\big(w(S)\big)^2
\]
for some $C>1/(1-(1+1/\sqrt{\lambda})\alpha)^2$.
Let $G$ be an $M\times N$ matrix with iid $N(0,1)$ entries.
Then $\Phi:=(GG^*)^{-1/2}G$ satisfies the $(\rho,\alpha)$-robust width property over $B_\sharp$ with probability $\geq1-4e^{-cN}$ for some $c=c(\lambda,\alpha,C)$.
\end{proposition}

For the sake of a familiar example, we consider the case where $\|\cdot\|_\sharp=\|\cdot\|_1$.
The following estimates $w(S)$ in this case (this is essentially accomplished in~\cite{MendelsonPT:05}):

\begin{lemma}
\label{lemma.L1 width}
There exists an absolute constant $c$ such that
\[
w(\sqrt{J}B_1\cap \mathbb{S}^{N-1})\leq c\sqrt{J\log(cN/J)}
\]
for every positive integer $J$.
\end{lemma}

\begin{proof}
For any fixed $g\in\mathbb{R}^N$, the $x\in\sqrt{J}B_1\cap \mathbb{S}^{N-1}$ which maximizes $\langle x,g\rangle$ has the same sign pattern as $g$ and the same order of entry sizes, i.e., $|g[i]|\geq|g[j]|$ if and only if $|x[i]|\geq|x[j]|$.
As such, we may assume without loss of generality that $g$ and $x$ have all nonnegative entries in nonincreasing order.
Then
\[
\langle x,g\rangle
=\langle x_J,g_J\rangle+\langle x-x_J,g-g_J\rangle
\leq\|x_J\|_2\|g_J\|_2+\sum_{i>J}x[i]g[i].
\]
Note that $\|x_J\|_2\leq\|x\|_2=1$ and $g[i]\leq g[J]\leq\|g_J\|_2/\sqrt{J}$ for every $i>J$, and so
\[
\langle x,g\rangle
\leq \|g_J\|_2+\frac{1}{\sqrt{J}}\|g_J\|_2\sum_{i>J}x[i]
\leq 2\|g_J\|_2,
\]
where the last step uses the fact that $x\in \sqrt{J}B_1$.
Next, we note that 
\[
\sup_{x\in\Sigma_J\cap B_2}\langle x,g\rangle
=\bigg\langle \frac{g_J}{\|g_J\|_2},g\bigg\rangle
=\|g_J\|_2.
\]
Letting $g$ have iid $N(0,1)$ entries, we then have
\[
w(\sqrt{J}B_1\cap \mathbb{S}^{N-1})
=\mathbb{E}\sup_{x\in \sqrt{J}B_1\cap\mathbb{S}^{N-1}}\langle x,g\rangle
\leq2\mathbb{E}\|g_J\|_2
=2w(\Sigma_J\cap B_2).
\]
At this point, we appeal to Lemma~3.3 in~\cite{MendelsonPT:05} (equivalently, Lemma~4.4 in~\cite{RudelsonV:08}), which gives
\[
w(\Sigma_J\cap B_2)
\leq c\sqrt{J\log(cN/J)}
\]
for some absolute constant $c$.
\end{proof}

Overall, if $M=\lambda N$ and we take an integer $J$ such that $M\geq5c^2J\log(cN/J)$, then Proposition~\ref{thm.main random result} and Lemma~\ref{lemma.L1 width} together imply that $\Phi:=(GG^*)^{-1/2}G$ satisfies $(\rho,\alpha)$-RWP over $B_1$ with high probability, where
\[
\rho:=\frac{1}{\sqrt{J}},
\qquad
\alpha
:=\frac{1}{2(1+1/\sqrt{\lambda})}.
\]
For the sake of comparison, consider $\Psi:=(1/\sqrt{\lambda})\Phi$.
(This random matrix is known to be a \textit{Johnson--Lindenstrauss projection}~\cite{DasguptaG:03}, and so it satisfies RIP with high probability~\cite{BaraniukDDW:08}, which in turn implies stable and robust compressed sensing~\cite{Candes:08}.)
Then $\Psi$ satisfies $(\rho,\alpha/\sqrt{\lambda})$-RWP with high probability.
Taking $\mathcal{A}=\Sigma_K$ with $K\leq J/16$, Theorem~\ref{thm.width-recovery} then gives that $\ell_1$ minimization produces an estimate $x^\star$ of $x^\natural$ from noisy measurements $\Phi x^\natural+e$ with $\|e\|_2\leq\epsilon$ such that
\[
\|x^\star-x^\natural\|_2
\leq\frac{1}{\sqrt{K}}\|x^\natural-x^\natural_K\|_1+4(1+\sqrt{\lambda})\epsilon.
\]
Note that $4(1+\sqrt{\lambda})\leq8$, and so these constants are quite small, even though we have not optimized them.

\subsection{Bowling scheme}

In the previous subsection, we were rather restrictive in our choice of sensing operators.
By contrast, this subsection will establish similar performance with a much larger class of random matrices.
The main tool here is the so-called \textit{bowling scheme}, coined by Tropp~\cite{Tropp:14}, which exploits the following lower bound for nonnegative empirical processes, due to Koltchinskii and Mendelson:

\begin{theorem}[Proposition~5.1 in~\cite{Tropp:14}, cf.\ Theorem~2.1 in~\cite{KoltchinskiiM:13}]
\label{thm.bowling}
Take a set $S\subseteq\mathbb{R}^N$.
Let $\varphi$ be a random vector in $\mathbb{R}^N$, and let $\Phi$ be an $M\times N$ matrix with rows $\{\varphi_i^\top\}_{i=1}^M$ which are independent copies of $\varphi^\top$.
Define
\[
Q_\xi(S;\varphi):=\inf_{x\in S}\operatorname{Pr}\Big(|\langle x,\varphi\rangle|\geq\xi\Big),
\qquad
W_M(S;\varphi):=\mathbb{E}\sup_{x\in S}\bigg\langle x,\frac{1}{\sqrt{M}}\sum_{i=1}^M\epsilon_i\varphi_i\bigg\rangle,
\]
where $\{\epsilon_i\}_{i=1}^M$ are independent random variables which take values uniformly over $\{\pm1\}$ and are independent from everything else.
Then for any $\xi>0$ and $t>0$, we have
\[
\inf_{x\in S}\|\Phi x\|_2
\geq\xi\sqrt{M}Q_{2\xi}(S;\varphi)-2W_M(S;\varphi)-\xi t
\]
with probability $\geq1-e^{-t^2/2}$.
\end{theorem}

As an example of how Theorem~\ref{thm.bowling} might be applied, consider the case where $\varphi$ has distribution $N(0,\Sigma)$.
We will take $\sigma_\mathrm{max}^2$ and $\sigma_\mathrm{min}^2$ to denote the largest and smallest eigenvalues of $\Sigma$, respectively.
First, we seek a lower bound on $Q_{2\xi}(S;\varphi)$.
We will exploit the fact that $\varphi$ has the same distribution as $\Sigma^{1/2}g$, where $g$ has distribution $N(0,I)$, and that
\[
\langle x,\Sigma^{1/2}g\rangle
=\langle \Sigma^{1/2}x,g\rangle
=\|\Sigma^{1/2}x\|_2\bigg\langle \frac{\Sigma^{1/2}x}{\|\Sigma^{1/2}x\|_2},g\bigg\rangle.
\]
Indeed, taking $Z$ to have distribution $N(0,1)$, then for any $x\in\mathbb{S}^{N-1}$, we have
\[
\operatorname{Pr}\Big(|\langle x,\varphi\rangle|\geq\xi\Big)
=\operatorname{Pr}\Big(|Z|\geq\xi/\|\Sigma^{1/2}x\|_2\Big)
\geq\operatorname{Pr}\Big(|Z|\geq\xi/\sigma_\mathrm{min}\Big)
\geq\frac{\sigma_\mathrm{min}}{\xi}\cdot\frac{1}{\sqrt{2\pi}}e^{-\xi^2/2\sigma_\mathrm{min}^2},
\]
where the last step assumes $\xi/\sigma_\mathrm{min}\geq1$.
Next, we pursue an upper bound on $W_M(S;\varphi)$.
For this, we first note that
\[
\operatorname{Pr}\Big(|\langle x,\varphi\rangle|\geq\xi\Big)
=\operatorname{Pr}\Big(|Z|\geq\xi/\|\Sigma^{1/2}x\|_2\Big)
\leq\operatorname{Pr}\Big(|Z|\geq\xi/\sigma_\mathrm{max}\Big)
\leq e^{-\xi^2/2\sigma_\mathrm{max}^2}
\]
for any $x\in\mathbb{S}^{N-1}$.
Furthermore, $\varphi$ has the same distribution as $\frac{1}{\sqrt{M}}\sum_{i=1}^M\epsilon_i\varphi_i$, and so
\[
\operatorname{Pr}\Bigg(\bigg|\bigg\langle u-v,\frac{1}{\sqrt{M}}\sum_{i=1}^M\epsilon_i\varphi_i\bigg\rangle\bigg|\geq\xi\Bigg)
\leq e^{-\xi^2/2\sigma_\mathrm{max}^2\|u-v\|_2^2}
\qquad
\forall u,v\in\mathbb{R}^N.
\]
As such, $\varphi$ satisfies the hypothesis of the generic chaining theorem (Theorem~1.2.6 in~\cite{Talagrand:05}), which, when combined with the majorizing measure theorem (Theorem~2.1.1 in~\cite{Talagrand:05}), gives
\begin{equation}
\label{eq.talagrand}
W_M(S;\varphi)
=\mathbb{E}\sup_{x\in S}\bigg\langle x,\frac{1}{\sqrt{M}}\sum_{i=1}^M\epsilon_i\varphi_i\bigg\rangle
\leq C\sigma_\mathrm{max}\cdot\mathbb{E}\sup_{x\in S}\langle x,g\rangle
=C\sigma_\mathrm{max}\cdot w(S).
\end{equation}
All together, we have
\begin{align*}
\inf_{x\in S}\|\Phi x\|_2
&\geq\xi\sqrt{M}Q_{2\xi}(S;\varphi)-2W_M(S;\varphi)-\xi t\\
&\geq\underbrace{\sqrt{M}\cdot(\sigma_\mathrm{min}/\sqrt{2\pi})e^{-\xi^2/2\sigma_\mathrm{min}^2}}_{a}-\underbrace{2C\sigma_\mathrm{max}\cdot w(S)}_{b}-\xi t.
\end{align*}
At this point, we pick $\xi=\sigma_\mathrm{min}$, $M$ such that $a=2b$, and $t$ such that $\xi t=(a-b)/2$ to get the following result:

\begin{proposition}
\label{prop.bowling 1}
Take $\rho>0$ and denote $S=\rho^{-1}B_\sharp\cap\mathbb{S}^{N-1}$.
Let $\varphi$ be distributed $N(0,\Sigma)$, and take $\sigma_\mathrm{max}^2$ and $\sigma_\mathrm{min}^2$ to denote the largest and smallest eigenvalues of $\Sigma$, respectively.
Set
\[
M=c_0\cdot\frac{\sigma_\mathrm{max}^2}{\sigma_\mathrm{min}^2}\cdot\big(w(S)\big)^2,
\qquad
\alpha=c_1\cdot\sigma_\mathrm{min}\sqrt{M},
\]
and let $\Phi$ be an $M\times N$ matrix whose rows are independent copies of $\varphi^\top$.
Then $\Phi$ satisfies the $(\rho,\alpha)$-robust width property over $B_\sharp$ with probability $\geq1-e^{c_2\sigma_\mathrm{min}^2M}$.
\end{proposition}

This result is essentially a special case of Theorem~6.3 in~\cite{Tropp:14}, which considers a more general notion of subgaussianity, and indeed, by this result, every matrix with iid subgaussian rows satisfies RWP.
However, we note that this result is suboptimal in certain regimes, in part thanks to the sledgehammers we applied in the estimate \eqref{eq.talagrand}.
To see the suboptimality here, consider the special case where $B_\sharp=B_1$.
Then for every $x\in S$, we have
\[
\langle x,\varphi\rangle
\leq\|x\|_1\|\varphi\|_\infty
\leq\rho^{-1}\|\varphi\|_\infty.
\]
Also, known results on maxima of Gaussian fields (e.g., equation~(2.13) in~\cite{LedouxT:91}) imply that
\[
\mathbb{E}\|\varphi\|_\infty
=\mathbb{E}\|\Sigma^{1/2}g\|_\infty
\leq3\sqrt{v\log N},
\]
where $v$ denotes the largest diagonal entry of $\Sigma$.
Putting things together, we have
\[
W_M(S;\varphi)
=\mathbb{E}\sup_{x\in S}\bigg\langle x,\frac{1}{\sqrt{M}}\sum_{i=1}^M\epsilon_i\varphi_i\bigg\rangle
=\mathbb{E}\sup_{x\in S}\langle x,\varphi\rangle
\leq \rho^{-1}\mathbb{E}\|\varphi\|_\infty
\leq3\rho^{-1}\sqrt{v\log N},
\]
which leads to the following result:

\begin{proposition}
\label{prop.bowling 2}
Take $\rho=1/\sqrt{J}$.
Let $\varphi$ be distributed $N(0,\Sigma)$, and take $v$ and $\sigma_\mathrm{min}^2$ to denote the largest diagonal entry and smallest eigenvalue of $\Sigma$, respectively.
Set
\[
M=c_0\cdot\frac{v}{\sigma_\mathrm{min}^2}\cdot J\log N,
\qquad
\alpha=c_1\cdot\sigma_\mathrm{min}\sqrt{M},
\]
and let $\Phi$ be an $M\times N$ matrix whose rows are independent copies of $\varphi^\top$.
Then $\Phi$ satisfies the $(\rho,\alpha)$-robust width property over $B_1$ with probability $\geq1-e^{c_2\sigma_\mathrm{min}^2M}$.
\end{proposition}

We note that this result could also have been deduced from Theorem~1 in~\cite{RaskuttiWY:10}, whose proof is a bit more technical.
Overall, this proposition exchanges $\sigma_\mathrm{max}^2$ for $v$ (which is necessarily smaller) and $\log(N/J)$ for $\log N$.
However, this is far from an even trade, as we illustrate in the following subsection.

\subsection{RWP does not imply RIP}

In this subsection, we consider a random matrix from Example~2 in~\cite{RaskuttiWY:10}.
This example will help to compare the performance of Propositions~\ref{prop.bowling 1} and~\ref{prop.bowling 2}, as well as provide a construction of an RWP matrix, no scaling of which satisfies RIP.

Pick $\Sigma:=\frac{1}{M}(I+\mathbf{1}_N\mathbf{1}_N^\top)$; we selected the $1/M$ scaling here so that $\|\varphi\|_2^2=\Theta(N/M)$ with high probability, as is typical for RIP matrices.
Then
\[
\sigma_\mathrm{max}^2=\frac{N+1}{M},
\qquad
\sigma_\mathrm{min}^2=\frac{1}{M},
\qquad
v=\frac{2}{M}.
\]
In this extreme case, Proposition~\ref{prop.bowling 1} uses $M\gg N$ rows to satisfy RWP, whereas Proposition~\ref{prop.bowling 2} uses only $O(J\log N)$ rows, and so the latter performs far better.
In either case, $\Phi$ is $(\rho,\alpha)$-RWP with $\rho=1/\sqrt{J}$ and $\alpha=O(1)$, mimicking the performance of an RIP matrix.
However, as we will show, this performance is \textit{logically independent} of the restricted isometry property, that is, the degree to which $\Phi$ satisfies RIP is insufficient to conclude that $\ell_1$ minimization exactly recovers all sparse signals, let alone with stability or robustness.

To see this, we start by following the logic of Example~2 in~\cite{RaskuttiWY:10}.
Take any $M\times J$ submatrix $\Phi_J$ of $\Phi$, and notice that the rows of $\Phi_J$ are iid with distribution $N(0,\Sigma_{JJ})$, where $\Sigma_{JJ}=\frac{1}{M}(I+\mathbf{1}_J\mathbf{1}_J^\top)$.
Take $u:=\mathbf{1}_J/\sqrt{J}$.
Then $\langle u,\varphi_J\rangle$ has distribution $N(0,\lambda_\mathrm{max}(\Sigma_{JJ}))$, and so $\|\Phi_J u\|_2^2/\lambda_\mathrm{max}(\Sigma_{JJ})$ has chi-squared distribution with $M$ degrees of freedom.
As such, Lemma~1 in~\cite{LaurentM:00} gives that 
\[
\operatorname{Pr}\bigg(\frac{\|\Phi_J u\|_2^2}{\lambda_\mathrm{max}(\Sigma_{JJ})}\leq\frac{M}{2}\bigg)\leq e^{-cM}.
\]
Similarly, for any unit vector $v$ which is orthogonal to $u$, we have that $\|\Phi_J v\|_2^2/\lambda_\mathrm{min}(\Sigma_{JJ})$ also has chi-squared distribution with $M$ degrees of freedom, and so the other bound of Lemma~1 in~\cite{LaurentM:00} gives
\[
\operatorname{Pr}\bigg(\frac{\|\Phi_J v\|_2^2}{\lambda_\mathrm{min}(\Sigma_{JJ})}\geq 2M\bigg)\leq e^{-cM}.
\]
Overall, we have that
\[
\frac{\lambda_\mathrm{max}(\Phi_J^*\Phi_J)}{\lambda_\mathrm{min}(\Phi_J^*\Phi_J)}
\geq\frac{\|\Phi_J u\|_2^2}{\|\Phi_J v\|_2^2}
>\frac{1}{4}\cdot\frac{\lambda_\mathrm{max}(\Sigma_{JJ})}{\lambda_\mathrm{min}(\Sigma_{JJ})}
=\frac{J+1}{4}
\]
with high probability.

Now suppose there is some scaling $\Psi=C\Phi$ such that $\Psi$ satisfies $(J,\delta)$-RIP.
Then
\[
\frac{1+\delta}{1-\delta}
\geq\frac{\lambda_\mathrm{max}(\Phi_J^*\Phi_J)}{\lambda_\mathrm{min}(\Phi_J^*\Phi_J)}
>\frac{J+1}{4},
\]
or equivalently, $\delta>(J-3)/(J+5)$.
At this point, we appeal to the recently proved Cai--Zhang threshold (Theorem~1 in~\cite{CaiZ:14}), which states that, whenever $t\geq 4/3$, $(tK,\delta)$-RIP implies exact recovery of all $K$-sparse signals by $\ell_1$ minimization if and only if $\delta<\sqrt{(t-1)/t}$; here, we are using the traditional ``with squares'' version of RIP.
As such, for $\Psi$, RIP guarantees the recovery of all $K$-sparse signals only if
\[
\frac{tK-3}{tK+5}
<\sqrt{\frac{t-1}{t}}.
\]
However, in order for this to hold for any $t$, a bit of algebraic manipulation reveals that we must have $K\leq 25$, a far cry from the RWP-based guarantee, which allows $K=\Omega(M/\log N)$.

\section{Discussion}

This paper establishes that in many cases, uniformly stable and robust compressed sensing is equivalent to having the sensing operator satisfy the robust width property (RWP).
We focused on the reconstruction algorithm denoted in the introduction by $\Delta_{\sharp,\Phi,\epsilon}$, but it would be interesting to consider other algorithms.
For example, the Lasso~\cite{Tibshirani:96} and the Dantzig selector~\cite{CandesT:07} are popular alternatives in the statistics community.
The restricted isometry property (RIP) is known to provide reconstruction guarantees for a wide variety of algorithms, but does RWP share this ubiquity, or is it optimized solely for $\Delta_{\sharp,\Phi,\epsilon}$?

We note that recently, ideas from geometric functional analysis have also been very successful in producing non-uniform compressed sensing guarantees~\cite{AmelunxenLMT:14,ChandrasekaranRPW:12,Tropp:14}.
In this regime, one is concerned with a Gaussian width associated with the descent cone at the signal $x^\natural$ instead of a dilated version of the entire $B_\sharp$ ball.
In either case, the Gaussian width of interest is the expected value of a random variable $\sup_{x\in S}\langle x,g\rangle$ for some fixed subset $S$ of the unit sphere.
Notice that this supremum can instead be taken over the convex hull of $S$, and so for every instance of $g$, one may efficiently compute $\sup_{x\in S}\langle x,g\rangle$ as a convex program.
As such, the desired expected value of this random variable can be efficiently estimated from a random sample.
The computational efficiency of this estimation is not terribly surprising in the non-uniform case, since one can alternatively attempt $\sharp$-norm minimization with a fixed $x^\natural$ and empirically estimate the probability of reconstruction.
This is a bit more surprising in the uniform case since for any fixed matrix, certifying a uniform compressed sensing guarantee is known to be NP-hard~\cite{BandeiraDMS:13,TillmannP:14}.
Of course, there is no contradiction here since (when combined with Proposition~\ref{thm.main random result}, Proposition~\ref{prop.bowling 1}, or more generally Theorem~6.3 in~\cite{Tropp:14}) this randomized algorithm merely certifies a uniform guarantee for most instances of a random matrix distribution.
Still, the proposed numerical scheme may be particularly useful in cases where the Gaussian width of $\rho^{-1}B_\sharp \cap\mathbb{S}$ is cumbersome to estimate analytically.

One interesting line of research in compressed sensing has been to find an assortment of random matrices (each structured for a given application, say) that satisfy RIP~\cite{BandeiraFMM:14,KrahmerMR:14,NelsonPW:14,Rauhut:10,RudelsonV:08}.
In this spirit, the previous section showed how the bowling scheme can be leveraged to demonstrate RWP for matrices with iid subgaussian rows.
We note that the bowling scheme (as described in~\cite{Tropp:14} in full detail) is actually capable of analyzing a much broader class of random matrices, though it is limited by the weaknesses of Theorem~\ref{thm.bowling}.
In particular, the bowling scheme requires $Q_\xi(S;\varphi)$ to be bounded away from zero, but this can be small when the distribution of $\varphi$ is ``spiky,'' e.g., when $\varphi$ is drawn uniformly from the rows of a discrete Fourier transform.
As such, depending on the measurement constraints of a given application, alternatives to the bowling scheme are desired.
Along these lines, Koltchinskii and Mendelson provide an alternative estimate of $\inf_{x\in S}\|\Phi x\|_2$ which depends on the VC dimension of a certain family of sets determined by $S$ (see Theorem~2.5 in~\cite{KoltchinskiiM:13}).
For the sake of a target, we pose the following analog to Problem~3.2 in~\cite{RudelsonV:08}:

\begin{problem}
What is the smallest $M=M(N,\rho,\alpha,\epsilon)$ such that drawing $M$ independent rows uniformly from the $N\times N$ discrete Fourier transform matrix produces a random matrix which satisfies the $(\rho,\alpha)$-robust width property over $B_1$ with probability $\geq1-\epsilon$?
\end{problem}

As a benchmark, it is known~\cite{CheraghchiGV:13} that taking
\[
M\geq\frac{C\log(1/\epsilon)}{\delta^2}K\log^3K\log N
\]
ensures that the properly scaled version of this random matrix satisfies $(K,\delta)$-RIP with probability $\geq1-\epsilon$, and so Theorem~\ref{thm.rip implies rwp} gives a corresponding upper bound on $M(N,\rho,\alpha,\epsilon)$; for the record, this uses the ``with squares'' version of RIP, but the difference in $M$ may be buried in the constant $C$.
Since RWP is strictly weaker than RIP, one might anticipate an improvement from an RIP-free approach.

\section*{Acknowledgements}

The original idea for this paper was conceived over mimosas in Pete Casazza's basement; we thank Pete for his hospitality and friendship.
This work was supported by NSF Grant No.\ DMS-1321779.
The views expressed in this article are those of the authors and do not reflect the official policy or position
of the United States Air Force, Department of Defense, or the U.S.\ Government.

\section*{Appendix}

\begin{proposition}
Let $\mathcal{D}$ denote the descent cone of $\|\cdot\|_\sharp$ at some nonzero $a\in\mathcal{H}$.
Take $v$ such that $\|v\|_\sharp=\|a\|_\sharp$ and $v-a\in\overline{\mathcal{D}^c}$, where $\overline{\mathcal{D}^c}$ denotes the topological closure of the set complement of $\mathcal{D}$.
Then
\[
\|a+x\|_\sharp
=\|a\|_\sharp+\|x\|_\sharp
\]
for every $x=cv$ with $c\geq0$.
\end{proposition}

\begin{proof}
Notice that
\[
\overline{\mathcal{D}^c}
=\overline{\bigcap_{t>0}\{y:\|a+ty\|_\sharp>\|a\|_\sharp\}}
\subseteq\bigcap_{t>0}\overline{\{y:\|a+ty\|_\sharp>\|a\|_\sharp\}}
=\bigcap_{t>0}\{y:\|a+ty\|_\sharp\geq\|a\|_\sharp\}.
\]
As such, $v-a\in\overline{\mathcal{D}^c}$ implies that
\[
\|a+t(v-a)\|_\sharp
\geq\|a\|_\sharp
\qquad
\forall t\geq0.
\]
Also, for every $t\in[0,1]$, convexity implies
\[
\|a+t(v-a)\|_\sharp
=\|(1-t)a+tv\|_\sharp
\leq(1-t)\|a\|_\sharp+t\|v\|_\sharp
=\|a\|_\sharp.
\]
Combining the last two displays then gives
\[
\|(1-t)a+tv\|_\sharp
=\|a\|_\sharp
\qquad
\forall t\in[0,1].
\]
With this, we get
\[
\|a+x\|_\sharp
=(1+c)\bigg\|\frac{1}{1+c}\cdot a+\frac{c}{1+c}\cdot v\bigg\|_\sharp
=(1+c)\|a\|_\sharp
=\|a\|_\sharp+c\|v\|_\sharp
=\|a\|_\sharp+\|x\|_\sharp,
\]
as desired.
\end{proof}

\end{document}